%% file: WCNC_2019.tex
\documentclass[conference]{IEEEtran}
\IEEEoverridecommandlockouts
\usepackage{ifpdf}
\usepackage{cite}
\usepackage[pdftex]{graphicx}
\usepackage{amsmath}
\usepackage{algorithmic}
\usepackage{array}
\usepackage{subfigure}
\usepackage{caption}
\usepackage[caption=false,font=footnotesize]{subfig}
\usepackage{fixltx2e}
\usepackage{stfloats}
\usepackage{url}
\usepackage{amsthm}
\usepackage{amsmath}
\usepackage{amsfonts}
\usepackage{amssymb}
\usepackage[T1]{fontenc} 
\usepackage{amsmath}
\usepackage[cmintegrals]{newtxmath}
\usepackage{bm} 
\usepackage[all]{xy}
\usepackage{enumitem}
\usepackage{float}
\usepackage{amsmath}
\usepackage[usenames, dvipsnames]{color}
\usepackage{ifpdf}
\usepackage{cite}
\usepackage[pdftex]{graphicx}
\usepackage{amsmath}
\usepackage{algorithmic}
\usepackage{array}
\usepackage{mathtools}
\DeclarePairedDelimiter{\abs}{\lvert}{\rvert}
\usepackage{amsmath}
\usepackage{amsfonts}
\usepackage{amssymb}
\usepackage{caption}
\usepackage[caption=false,font=footnotesize]{subfig}
\usepackage{fixltx2e}
\usepackage{stfloats}
\usepackage{url}
\usepackage{mathtools}
\usepackage{amsthm}
\usepackage{amsmath}
\usepackage{amsfonts}
\usepackage{amssymb}
\usepackage[T1]{fontenc} 
\usepackage{amsmath}
\usepackage[cmintegrals]{newtxmath}
\usepackage{bm} 
\usepackage[all]{xy}
\usepackage{enumitem}
\usepackage{float}
\usepackage{amsmath}
\usepackage[dvipsnames]{xcolor}
\usepackage{multirow}
\usepackage{subfig}
\usepackage{multicol}
\usepackage{graphicx}
\usepackage[caption=false]{subfig}
\usepackage{amssymb}
\usepackage{amsmath}
\usepackage{commath}
\usepackage{graphicx,bm}
\usepackage{verbatim}

\usepackage{diagbox}

\theoremstyle{definition}

\newtheorem{thm}{Theorem}

\newtheorem{lem}{Lemma}

\newtheorem{define}{Definition}
\newcommand{\cov}{\textrm{Cov}}

\newcommand{\no}{\nonumber}

\begin{document}
	\title{Asymptotic Loss in Privacy due to Dependency in Gaussian Traces}	
\author{\IEEEauthorblockN{Nazanin Takbiri}
	\IEEEauthorblockA{Electrical and\\Computer Engineering\\
		UMass-Amherst\\
		ntakbiri@umass.edu}
	\and
	\IEEEauthorblockN{Ramin Soltani}
	\IEEEauthorblockA{Electrical and\\Computer Engineering\\
		UMass-Amherst\\
		soltani@ecs.umass.edu}
	\and
	\IEEEauthorblockN{Dennis L. Goeckel}
	\IEEEauthorblockA{Electrical and\\Computer Engineering\\
		UMass-Amherst\\
		goeckel@ecs.umass.edu}
	\and
	\IEEEauthorblockN{Amir Houmansadr}
	\IEEEauthorblockA{Information and \\Computer Sciences\\
		UMass-Amherst\\
		amir@cs.umass.edu}
	\and
	\IEEEauthorblockN{Hossein Pishro-Nik}
	\IEEEauthorblockA{Electrical and\\Computer Engineering\\
		UMass-Amherst\\
		pishro@ecs.umass.edu\thanks{This work was supported by National Science Foundation under grants CCF--1421957 and CNS-1739462.}}
}
\maketitle

\begin{abstract}	
The rapid growth of the Internet of Things (IoT) necessitates employing privacy-preserving techniques to protect users' sensitive information. Even when user traces are anonymized, statistical matching can be employed to infer sensitive information. In our previous work, we have established the privacy requirements for the case that the user traces are instantiations of discrete random variables and the adversary knows only the structure of the dependency graph, i.e., whether each pair of users is connected. In this paper, we consider the case where data traces are instantiations of Gaussian random variables and the adversary knows not only the structure of the graph but also the pairwise correlation coefficients.  We establish the requirements on anonymization to thwart such statistical matching, which demonstrate the significant degree to which knowledge of the pairwise correlation coefficients further significantly aids the adversary in breaking user anonymity.
\end{abstract}



\input{introduction}
\input{framework-4}



\input{conclusion}

\appendices

\bibliographystyle{IEEEtran}
\bibliography{REF,dennis_privacy,Refrences}
%
%
%
%


\end{document}

%% file: introduction.tex
\section{Introduction}
\label{intro}
\IEEEPARstart{T}{he} Internet of Things (IoT) enables users to share and access information on a large scale and provides many benefits to individuals (e.g., smart homes, healthcare) and industries (e.g., digital tracking, data collection, disaster management). However, such benefits are provided by tuning the system to user characteristics based on potentially sensitive information about their activities. Thus, the use of IoT comes with a significant threat to users' privacy: leakage of sensitive information. 

Two main privacy-preserving techniques are anonymization~\cite{ma2009location,shokri2011quantifying2,soltani2017towards} and obfuscation~\cite{duckham2005formal,ardagna2011obfuscation}, where the former is hiding the mapping between data and the users by replacing the identification fields of users with pseudonyms, and the latter is perturbing the user data such that the adversary observes false but plausible data. Although these methods have been addressed widely, statistical inference methods can be applied to them to break the privacy of the users~\cite{adam_twireless,3ukil2014iot, FTC2015,0Quest2016}. Furthermore, achieving privacy using these methods comes with a cost: reducing the utility of the system for the users. Hence, it is crucial to consider the trade-off between privacy and utility when employing privacy-preserving techniques, and to seek to achieve privacy with minimal loss of functionality and usability~\cite{loukides2012utility,lee2017utility,batet2013utility}.
Despite the growing interest in IoT privacy~\cite{ukil2014iot,lin2016iot}, previous works do not offer theoretical guarantees on the trade-off between privacy and utility. The works of Shokri et al.~\cite{shokri2011quantifying,shokri2011quantifying2,shokri2012protecting} and Ma et al.~\cite{ma2009location} provide significant advances in the quantitative analyses of privacy; however, in contrast to these prior works, we take a foundational approach to understand the theoretical limits.

In \cite{nazanin_ISIT2017,ISIT17-longversion1, KeConferance2017,ciss2017}, the data traces of different users are modeled as independent, and the asymptotic limits of user privacy are presented for the case when both anonymization and obfuscation are applied to users' time series of data. In ~\cite{nazanin_ISIT2017,ISIT17-longversion1}, each user's data trace is governed by: 1) independent identically distributed (i.i.d.) samples of a Multinoulli distribution (generalized Bernoulli distribution) with $r$ possibilities for each data point.; or, 2) Markov chain samples of a Multinoulli distribution, where each user's data
samples are governed by a Markov chain with $r$ states. In ~\cite{KeConferance2017}, the case of independent users with Gaussian traces was addressed. 
However, the data traces of different users are dependent in many applications (e.g,. friends, relatives), and the adversary can potentially exploit such. In~\cite{nazanin_ISIT2018, ISIT18-longversion}, we extended the results of \cite{nazanin_ISIT2017,ISIT17-longversion1} to the case where users are dependent and the adversary knows only the structure of the association graph, i.e., whether each pair of users are linked. As expected, the knowledge of the dependency graph results in a significant degradation in privacy \cite{nazanin_ISIT2017,ISIT17-longversion1}.

In this paper, we turn our attention to the case where the trace of each user consists of identically distributed Gaussian random variables that are independent over time, but there is dependency between the samples of different users at each point in time. The adversary knows not just the dependency graph but also the degree to which the data of different users are correlated, and thus the adversary knows the joint probability distribution function (PDF) of the data generated by all of the users. To preserve the privacy of users, anonymization is employed, i.e., the mapping between users and data sequences is randomly permuted for each set of $m$ consequent users' data. We derive the minimum number $m$ for the adversary's observations per user that ensures privacy, with respect to the number $n$ of users and the size $s$ of the sub-graph to which the user belongs. 

The rest of the paper is organized as follows. In Section~\ref{sec:framework}, we present the framework: system model, metrics, and definitions. Then, we present the construction and analysis in Section~\ref{anon}. In Section~\ref{discussion}, we discuss how inter-user dependency affects system privacy, and in Section~\ref{conclusion}, we conclude from the results.

%% file: framework-4.tex
\section{Framework}
\label{sec:framework}
Consider a system with $n$ users. Denote by $X_u(k)$ the data point of user $u$ at time $k$, and by $\textbf{X}_u$ the $m \times 1$ vector containing the data points of user $u$, 
\[ \textbf{X}_u =\left[{X}_{u}(1), {X}_{u}(2), \cdots, {X}_{u}(m)\right]^T, \ \ \ \ u\in \{1, 2, \cdots, n\}.\]

To preserve the privacy of the users, anonymization is employed, with pseudonyms that are changed every $m$ samples, i.e., the mapping between users and data sequences is randomly permuted every $m$ samples. As shown in Figure \ref{fig:xyz}, denote by ${Y}_u(k)$ the output of the anonymizer, which we term as the ``reported data point'' of user $u$ at time $k$. The permuted version of $\textbf{X}_u$ is
\[\textbf{Y}_u =\left[{Y}_{u}(1), {Y}_{u}(2), \cdots, {Y}_{u}(m)\right]^T, \ \ \ \ u\in \{1, 2, \cdots, n\},\]
\begin{figure}[h]
	\centering
	\includegraphics[width =0.6\linewidth]{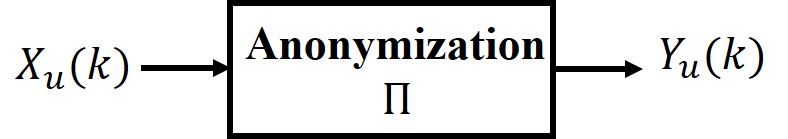}
	\caption{Applying anonymization to the data point of user $u$ at time $k$. $X_u(k)$ denotes the actual data point of user $u$ at time $k$, and $Y_u(k)$ denotes the reported data point of user $u$ at time $k$.}
	\label{fig:xyz}
\end{figure}
The anonymization technique can be modeled by a random permutation function $(\Pi)$ on the set of $n$ users. Then,
$\textbf{Y}_u=\textbf{X}_{\Pi^{-1}(u)}$, $ \textbf{Y}_{\Pi(u)}=\textbf{X}_{u}.$

There exists an adversary who wishes to break the anonymity and thus privacy of the users. He observes $\textbf{Y}_1,\textbf{Y}_2,\ldots,\textbf{Y}_n$ which are the reported data points of $n$ users at times $1,2,\ldots,m$, and performs statistical analysis to estimate the users' actual data points. 
\subsection{Models and Metrics}
{\textit{Data Points Model:}}
Data points are independent and identically distributed (i.i.d.) with respect to time, i.e., $\forall k,k' \in \{1,2,\ldots,m\}$, $k \neq k'$, ${X}_{u}(k)$ is independent of ${X}_{u}(k')$. At time $k$, the vector of user data points is drawn from a multivariate normal distribution; that is, 
$$[X_1(k), X_2(k),\ldots,X_n(k)] \sim \mathcal{N} \left(\boldsymbol{\mu}, \boldsymbol{\Sigma}\right),$$
where $\boldsymbol{\mu}=[\mu_1, \mu_2, \ldots,\mu_n]$ is the mean vector, $\boldsymbol{\Sigma}$ is the $n \times n$ covariance matrix, $\Sigma_{u,u'}=\sigma^2_{uu'}=\mu_{uu'}-\mu_{u}\mu_{u'}$ is the covariance between users $u$ and $u'$, and the variances of the user data points are equal for all users ($\Sigma_{u,u'}=\sigma^2$). Following our previous work \cite{ISIT18-longversion}, the parameters of the distribution governing users' behavior are in turn drawn randomly. In particular, we assume the means $\mu_1, \mu_2, \cdots, \mu_n$ are finite and are drawn independently from a continuous distribution $f_\mu(x)$, where for all $x$ in the support of $f_\mu(x)$
\begin{align}
0 <\delta_{1} <f_\mu(x)<\delta_{2}< \infty,
\label{case1}
\end{align}
and the correlations $\mu_{uu'}, u=1,2, \cdots, n, u'=1,2, \cdots, n, u \neq u'$ are finite, and, when two users are correlated, are drawn independently from a continuous distribution $g_\mu(x)$, where for all $x$ in the support of $g_\mu(x)$
\begin{equation}
0<\delta_{1}\leq g_\mu(x) \leq\delta_{2} < \infty.
\label{case2}
\end{equation}
Although it will not affect our results, we note in passing that the Cauchy-Schwarz inequality shows that $\sigma^2$ implies an upper bound on the support of $g_{\mu}(x)$.

{\textit{Association Graph:}}
The dependencies between users are modeled by an association graph in which two users are connected if they are dependent. Denote by $G(V,E)$ the association graph where $V$ is the set of nodes $(|V|=n)$ and $E$ is the set of edges. Also, denote by $\rho_{uu'}=\frac{\mu_{uu'}-\mu_u \mu_{u'}}{\sigma^2}$ the correlation coefficient between users $u$ and $u'$. Observe
$(u,u') \in E \text{ iff } \rho_{uu'} \neq 0.$
Similar to~\cite{ISIT18-longversion}, the association graph consists of disjoint subgraphs $G_1(V_1, E_1), G_2(V_2, E_2), \ldots, G_f(V_f, E_f)$, where each subgraph $G_j$ is connected and refers to a group of ``friends'' or ``associates.'' Let $s_j$ denotes the number of nodes in $G_j(V_j, E_j)$, i.e., $s_j=|V_j|$.
\begin{figure}
	\centering
	\includegraphics[width=.4\linewidth]{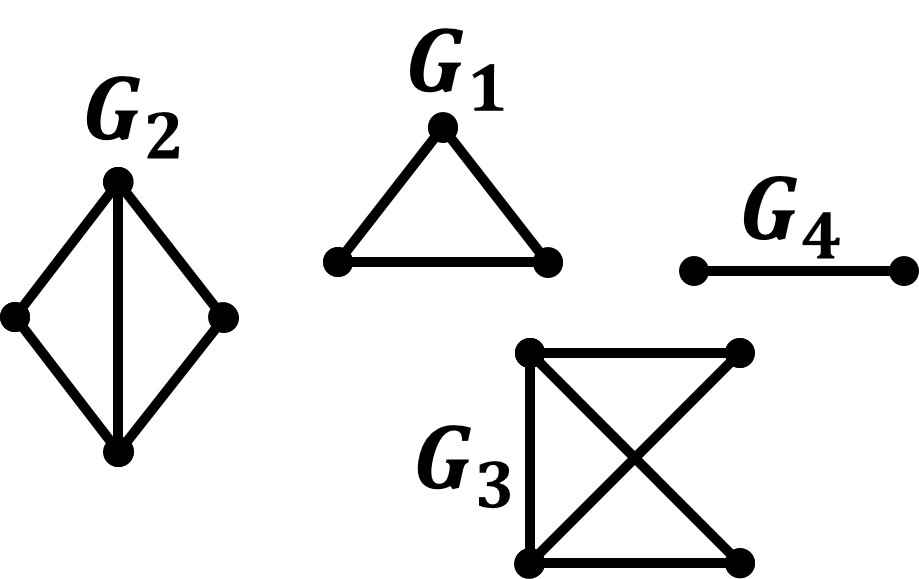}
	\centering
	\caption{The association graph consists of disjoint subgraphs $(G_j)$, where $G_j$ is a connected graph on $s_j$ vertices.}
	\label{fig:graph}
\end{figure}


{\textit{Adversary Model:}} 
The adversary knows the multivariate normal distribution from which the data points of users are drawn. Therefore, in contrast to~\cite{ISIT18-longversion}, the adversary knows both the structure of the association graph $G(\mathcal{V},E)$ as well as the correlation coefficients $(\rho_{uu'})$ for each pair of users $(u,u') \in E$. The adversary also knows the anonymization mechanism; however, they don't not know the realization of the random permutation function.

The situation in which the user has no privacy is defined as follows~\cite{nazanin_ISIT2017}:
\begin{define}
User $u$ has \emph{no privacy} at time $k$, if and only if there exists an algorithm for the adversary to estimate $X_u(k)$ perfectly as $n$ goes to infinity. In other words, as $n \rightarrow \infty$,
\begin{align}
\no \forall k\in \mathbb{N}, \ \ \ \mathbb{P}_e(u)\triangleq \mathbb{P}\left(\widetilde{X_u(k)} \neq X_u(k)\right)\rightarrow 0,
\end{align}
where $\widetilde{X_u(k)}$ is the estimated value of $X_u(k)$ by the adversary.
\end{define}

\section{Impact of Dependency on Privacy Using Anonymization}
\label{anon}
Here, we consider to what extent inter-user dependency limits privacy in the case users' data points are governed by a Gaussian distribution. 

The adversary knows the structure of the statistical dependency and seeks to use it to match users to their data sequences. First, the adversary must figure out the dependency in the data traces. Thus, we consider the ability of the adversary to fully reconstruct the structure of the association graph of the anonymized version of the data with arbitrarily small error probability.

\begin{lem}
If for any $\lambda>0$, the adversary obtains $m=n^{\lambda}$ anonymized observations, they can reconstruct $\widetilde{G}=\widetilde{G}(\widetilde{\mathcal{V}}, \widetilde{E})$, where $\widetilde{\mathcal{V}}=\{\Pi(u):u \in \mathcal{V}\}=\mathcal{V}$, such that with high probability, for all $u, u' \in \mathcal{V}$; $ (u,u')\in E$ iff $\left(\Pi(u),\Pi(u')\right)\in \widetilde{E}$. We write this statement as $\mathbb{P}(\widetilde{E}=E)\to 1$.
\label{lemma1}
\end{lem}

\begin{proof}
From the observations, the adversary can calculate the empirical covariance for each pair of user $u$ and user $u'$, 
\begin{align}
\widetilde{\cov_{uu'}}&=\frac{S_{uu'}}{m}-\frac{S_{u}}{m}\frac{S_{u'}}{m},
\label{cov}
\end{align}
where
\begin{align}
	S_{u}=\sum_{k=1}^{m}X_u(k),\ \ S_{uu'}&=\sum_{k=1}^{m}X_u(k)X_{u'}(k).
	\label{S}
\end{align}	
We claim for $m=n^{\lambda}$, and large enough $n$,
\begin{itemize}
	\item $|\widetilde{\cov_{uu'}}|\leq{m^{-\frac{1}{5}}}$, iff $(u,u')\notin \widetilde{E},$
	\item $|\widetilde{\cov_{uu'}}|>{m^{-\frac{1}{5}}}$, iff $(u,u')\in \widetilde{E},$
\end{itemize}
In other words, we show $P(\widetilde{E}=E)\to 1$ as $n \to \infty$.

Now, define
\begin{align}
\no \mathcal{J}_{uu'}=\bigg\{\bigg{|}\widetilde{\cov_{uu'}}-\left(\mu_{uu'}-\mu_u\mu_{u'}\right)\bigg{|}\geq 2\theta\bigg\};
\end{align}
thus, we have
\begin{align}
\no \mathbb{P}\left(\mathcal{J}_{uu'}\right)&=\mathbb{P}\left(\bigg{|}\left(\frac{S_{uu'}}{m}-\mu_{uu'}\right)- \left(\frac{S_{u}}{m}\frac{S_{u'}}{m}-\mu_{u}\mu_{u'}\right)\bigg{|}\geq 2\theta\right)\\
\no &\hspace{-0.4 in}\leq\mathbb{P}\left(\bigg{|}\frac{S_{uu'}}{m}-\mu_{uu'}\bigg{|}+ \bigg{|}\frac{S_{u}}{m}\frac{S_{u'}}{m}-\mu_{u}\mu_{u'}\bigg{|}\geq 2\theta\right)\\
\no &\hspace{-0.4in}\leq \mathbb{P}\left(\bigg{\{}\bigg{|}\frac{S_{uu'}}{m}-\mu_{uu'}\bigg{|}\geq \theta\bigg{\}} \bigcup \bigg{\{}\bigg{|}\frac{S_{u}}{m}\frac{S_{u'}}{m}-\mu_{u}\mu_{u'}\bigg{|}\geq \theta\bigg{\}} \right)\\
&\hspace{-0.4 in}\leq \mathbb{P}\left(\bigg{|}\frac{S_{uu'}}{m}-\mu_{uu'}\bigg{|}\geq \theta \right)+\mathbb{P}\left( \bigg{|}\frac{S_{u}}{m}\frac{S_{u'}}{m}-\mu_{u}\mu_{u'}\bigg{|}\geq \theta\right),\ \
\label{J_uu'}
\end{align}
where the first inequality follows from the fact that $\abs{a-b} \leq \abs{a}+\abs{b}$, and as a result, $\mathbb{P}\left(\ \abs{a-b} \geq 2\theta\right) \leq \mathbb{P}\left(\ \abs{a}+\abs{b} \geq 2\theta\right)$. The union bound yields the third inequality.

Note that we have
\begin{align}
\nonumber \mathbb{P}\left(\bigg{|}\frac{S_{uu'}}{m}-\mu_{uu'}\bigg{|}\geq \theta\bigg{|}\right)&\leq \frac{\mathbb{E}\left[\left(\sum_{k=1}^{m}\left(X_u(k)X_{u'}(k)-\mu_{uu'}\right)\right)^{\zeta}\right]}{\theta^{\zeta}m^{\zeta}},\\
 \nonumber &\hspace{-0.3 in}\leq \frac{\tau \mathbb{E}\left[\left({\sum_{k=1}^{m}\left(X_u(k)X_{u'}(k)-\mu_{uu'}\right)^2}\right)^{{\zeta/2}}\right]}{\theta^{\zeta}m^{\zeta}},\\
&\hspace{-0.3 in}\leq \frac{\tau \mathbb{E}\left[\left(\frac{\sum_{k=1}^{m}\left(X_u(k)X_{u'}(k)-\mu_{uu'}\right)^2}{m}\right)^{{\zeta/2}}\right]}{\theta^{\zeta}m^{\zeta/2}},\ \
 \label{eq:1} 
\end{align}
where the first and second steps follow from Chebyshev's inequality and the Marcinkiewicz-Zygmund inequality~\cite{chow2012probability}, respectively, and $\tau>0$ is a constant independent of $m$. Note that the Marcinkiewicz-Zygmund requires $\mathbb{E}[\left(X_u(k)X_{u'}(k)-\mu_{uu'}\right)^{\zeta}]<+\infty$ which follows from the Cauchy-Schwarz inequality and the fact that the $\zeta^{th}$ moments of $X_u(k)$ and $X_{u'}(k)$ are finite.

Consider the right-hand side (RHS) of~\eqref{eq:1}. Since $f(x)=x^{{\zeta/2}}$ is a convex function of $x$ when $x>0$, Jensen's inequality yields:
\begin{align}
\nonumber &\left(\frac{{\sum\limits_{k=1}^{m}\left(X_u(k)X_{u'}(k)-\mu_{uu'}\right)^2}}{m}\right)^{{\zeta/2}}\leq \frac{\sum\limits_{k=1}^{m}\left(X_u(k)X_{u'}(k)-\mu_{uu'}\right)^{\zeta}}{m}.
\end{align}
Consequently,~\eqref{eq:1} yields:
\begin{align}
\no \mathbb{P}\left(\bigg{|}\frac{S_{uu'}}{m}-\mu_{uu'}\bigg{|}\geq \theta\bigg{|}\right)&\leq \frac{\tau \mathbb{E}\left[\frac{\sum\limits_{k=1}^{m}\left(X_u(k)X_{u'}(k)-\mu_{uu'}\right)^{\zeta}}{m}\right]}{\theta^{\zeta}m^{\zeta/2}}\\
&=\frac{\tau \mathbb{E}\left[\left(X_u(k)X_{u'}(k)-\mu_{uu'}\right)^{\zeta}\right]}{\theta^{\zeta}m^{\zeta/2}}
\label{A_uu'}.
\end{align}
Note that $\mathbb{E}\left[\left(X_u(k)X_{u'}(k)-\mu_{uu'}\right)^{\zeta}\right]$ on the RHS of~\eqref{A_uu'} is finite for $0\leq \zeta < \infty$, following from the Cauchy-Schwarz inequality and the fact that the $\zeta^{th}$ moments of $X_u(k)$ and $X_{u'}(k)$ are finite.

Also, since $\frac{S_u}{m}-\mu$ has a zero-mean normal distribution with a variance equal to $\frac{\sigma^2}{m}$, we have
\begin{align} 
\mathbb{P}\left(\bigg{|}\frac{S_{u}}{m}-\mu_u\bigg{|}\geq \theta\right)= 1-\text{erf}\left({\frac{\sqrt{m} \theta}{\sqrt{2}\sigma}}\right) \leq e^{-\frac{m\theta^2 }{2\sigma^2}},
\label{A_u}
\end{align}
where the last step is true because $\text{erf}(x) \geq 1-e^{-x^2}$.
Now, if $\theta \to 0$, we have
\begin{align}
\no &\mathbb{P}\left( \bigg{|}\frac{S_{u}}{m}\frac{S_{u'}}{m}-\mu_{u}\mu_{u'}\bigg{|}\geq \theta\right) = \\
\no & \hspace{0.5 in} =\mathbb{P}\left( \bigg{|}\frac{S_{u}}{m}-\mu_{u}\bigg{|}\geq \theta'\right) \mathbb{P}\left( \bigg{|}\frac{S_{u'}}{m}-\mu_{u'}\bigg{|}\geq \theta'\right)\\
 &\hspace{0.5 in}=\left(1-\text{erf}\left({\frac{\sqrt{m} \theta'}{\sqrt{2}\sigma}}\right)\right)^2 \leq e^{-\frac{m\theta'^2 }{\sigma^2}},\ \
\label{S_uu'}
\end{align}
where $\theta'=\frac{\theta}{\mu_u+\mu_u'}$.

Let $m=n^{\lambda}$, $\theta=m^{-\frac{1}{4}}$, and $\zeta=\lceil{\frac{8}{\lambda}}\rceil$. By (\ref{J_uu'}), (\ref{A_uu'}), and (\ref{S_uu'}), the union bound yields 
\begin{align}
\no &\mathbb{P} \left(\bigcup\limits_{u=1}^n \bigcup\limits_{u'=1}^n \mathcal{J}_{uu'} 	\right) \\ 
\no &\leq \frac{\tau \mathbb{E}\left[\left(X_{u}(k)X_{u'}(k)-\mu_{uu'}\right)^{2s(s+1)}\right]}{n^2\left(n^{\frac{\lambda}{4}}\right)^{\lceil{\frac{8}{\lambda}}\rceil}} +n^2e^{\frac{1}{{(\mu_u+\mu_{u'})^2\sigma^2}}{n^{-\frac{\lambda}{2}}}}, \ \
\end{align}	
as a result, we can conclude as $n \to \infty$,
$$\mathbb{P} \left(\bigcup\limits_{u=1}^n \bigcup\limits_{u'=1}^n \mathcal{J}_{uu'} 	\right) \to 0.$$ 

Now, we can conclude with high probability, for all $(u,u') \notin E$ (which means $\rho_{uu'}= 0)$,
\begin{align}
\no |\widetilde{\cov_{uu'}}|\leq m^{-\frac{1}{5}},
\end{align}
as $n \to \infty.$
On the other hand, with high probability, for all $(u,u') \in E$ (which means $\mu_{uu'}-\mu_u\mu_{u'}\neq 0 $), 
\begin{align}
\no |\widetilde{\cov_{uu'}}|\geq m^{-\frac{1}{5}},
\end{align}
as $n \to \infty$.
Consequently, the adversary can reconstruct the association graph of the anonymized version of the data with arbitrarily small error probability. 
\end{proof}

Next, we demonstrate how the adversary can identify group $1$ among all of the groups. Note that this is the key step which speeds up the adversary's algorithm relative to the case where user traces are independent.
\begin{lem}
If for any $\alpha>0$, the adversary obtains $m=n^{\frac{4}{s(s+1)}+\alpha}$ anonymized observations and knows the structure of the association graph, they can identify group $1$ among all of the groups with arbitrarily small error probability.
\label{lemma2}
\end{lem}

\begin{proof}
Note that there are at most $\frac{n}{s}$ groups of size $s$ which we denote $1, 2, \cdots, \frac{n}{s}$. Without loss of generality, we assume the members of group $1$ are users $\{1,2, \cdots , s\}$.

By \eqref{S}, for all members of group $1$ ($u \in \{1, 2, \cdots, s\}$), the empirical mean $\widetilde{\mu_{\Pi{(u)}}}$ is:
\begin{align}
\widetilde{\mu_{\Pi(u)}}=\frac{S_{u}}{m}.\ \
\label{p-tilde2}
\end{align}
For $i \in \{1, 2, \cdots, s\}$, define vectors $\textbf{P}^{*}_i$ and $\widetilde{\textbf{P}^{*}_i}$ with length $s-i$:
\[\textbf{P}^{*}_i=[\mu_{(i)(i+1)}, \mu_{(i)( i+2)}, \cdots, \mu_{(i)(s)}],\]
\[\widetilde{\textbf{P}^{*}_i}=[\widetilde{\mu_{(i)(i+1)}}, \widetilde{\mu_{(i)( i+2)}}, \cdots,\widetilde {\mu_{(i)(s)}}],\]
and for $i=0$, define
\[\textbf{P}^{*}_0=[\mu_1, \mu_2, \cdots, \mu_s], \ \ \widetilde{\textbf{P}^{*}_0}=[\widetilde{\mu_1}, \widetilde{\mu_2}, \cdots,\widetilde {\mu_s}].\]
Also, define arrays $\textbf{P}^{(1)},\widetilde{\textbf{P}^{(1)}} \in\mathbb{R}^{\frac{s(s+1)}{2}}$ as:
\[\textbf{P}^{(1)}=[\textbf{P}^{*}_0, \textbf{P}^{*}_1, \cdots, \textbf{P}^{*}_s ], \ \ \widetilde{\textbf{P}^{(1)}}=[\widetilde{\textbf{P}^{*}_0}, \widetilde{\textbf{P}^{*}_1}, \cdots, \widetilde{\textbf{P}^{*}_s} ].\]
Let $\Pi_s$ be the set of all permutations on $s$ elements; for $\pi_s \in \Pi_s$, $\pi_s :\{1, 2, \cdots, s\}\to \{1, 2, \cdots, s\}$ is a one-to-one mapping. From~\cite[Equation 6]{ISIT18-longversion}, define
\begin{align}
\mathcal{D}\left(\mathbf{P}^{(1)},\widetilde{\mathbf{P}^{(1)}}\right)=\min\limits_{\pi_s \in \Pi_s}\left\{||\mathbf{P}^{(1)}-\widetilde{\mathbf{P}^{(1)}}_{\pi_s}||_{\infty}\right\}.
\end{align}
Next, we show when $m=cn^{\frac{4}{s(s+1)}+\alpha}$ and $n \to \infty$,
\begin{itemize}
	\item $\mathbb{P}\left(\mathcal{D}\left(\textbf{P}^{(1)}, \widetilde{\textbf{P}^{(1)}}\right)\leq \Delta_n\right) \to 1, $
	\item $\mathbb{P}\left(\bigcup\limits_{l=2}^{\frac{n}{s}} \left\{ \mathcal{D}\left(\textbf{P}^{(1)},\widetilde{\textbf{P}^{(l)}}\right)\leq \Delta_n \right\}\right) \to 0$ ,
\end{itemize}
where $\Delta_n= {n^{-\frac{2}{s(s+1)}-\frac{\alpha}{4}}}$.



First, we prove $\mathcal{D}\left(\textbf{P}^{(1)}, \widetilde{\textbf{P}^{(1)}}\right)\leq \Delta_n$ with high probability. Substituting $\theta= \Delta_n$ in (\ref{A_uu'}) and (\ref{A_u}) yields:
\begin{align}
\no \mathbb{P}\left(\bigg{|}\frac{S_{uu'}}{m}-\mu_{uu'}\bigg{|}\geq \Delta_n\right)&\leq \frac{\tau \mathbb{E}\left[\left(X_{u}(k)X_{u'}(k)-\mu_{uu'}\right)^{\zeta}\right]}{\Delta_n^{\zeta}m^{\zeta/2}}\\
&=\tau \mathbb{E}\left[\left(X_{u}(k)X_{u'}(k)-\mu_{uu'}\right)^{\zeta}\right]{n^{-\frac{\alpha}{4}\zeta }},
\label{eq10}
\end{align}
and
\begin{align}
\mathbb{P}\left( \bigg{|}\frac{S_u}{m}-\mu_{u}\bigg{|}\geq \Delta_n \right) &\leq e^{-\frac{m\Delta_n^{2}}{2\sigma^2}}= e^{-\frac{1}{2\sigma^2}n^{\frac{\alpha}{2}}}.
\label{eq11}
\end{align}
By the union bound,
\begin{align}
\no & \mathbb{P}\left(\mathcal{D}\left(\textbf{P}^{(1)},\widetilde{\textbf{P}^{(1)}}\right)\geq \Delta_n\right)\\
 \no &\hspace{0.0 in} \leq \sum\limits_{u=1}^{s} \mathbb{P}\left(\bigg{|}\frac{S_{u}}{m}-\mu_{u}\bigg{|}\geq \Delta_n\right)+\sum\limits_{u=1}^{s} \sum_{u'=u+1}^{s}\mathbb{P}\left(\bigg{|}\frac{S_{uu'}}{m}-\mu_{uu'}\bigg{|}\geq \Delta_n\right) \\
\no &\hspace{0.0in} = s\mathbb{P}\left(\bigg{|}\frac{S_{u}}{m}-\mu_{u}\bigg{|}\geq \Delta_n\right)+\frac{s(s-1)}{2} \mathbb{P}\left(\bigg{|}\frac{S_{uu'}}{m}-\mu_{uu'}\bigg{|}\geq \Delta_n\right) \\
\no &\hspace{0.0 in}\leq se^{-\frac{1}{2\sigma^2}n^{\frac{\alpha}{2}}} +\frac{s(s-1)}{2}\tau \mathbb{E}\left[\left(X_{u}(k)X_{u'}(k)-\mu_{uu'}\right)^{\zeta}\right]{n^{-\frac{\alpha}{4} \zeta}},\ \
\end{align}
consequently, as $n \to \infty,$
$$\mathbb{P}\left(\mathcal{D}\left(\textbf{P}^{(1)},\widetilde{\textbf{P}^{(1)}}\right)\leq \Delta_n\right) \to 1.$$

Next, we show $$\mathbb{P}\left(\bigcup\limits_{l=2}^{\frac{n}{s}} \left\{ \mathcal{D}\left(\textbf{P}^{(1)},\widetilde{\textbf{P}^{(l)}}\right)\leq \Delta_n \right\}\right) \to 0.$$ Note that by (\ref{case1}) and (\ref{case2}), for all groups other than group $1$,
\begin{align}
\no \mathbb{P}\left( \big{|}\textbf{P}^{(1)} - \textbf{P}^{(l)}\big{|} \leq 2 \Delta_n\right) &\leq (4 \Delta_n)^{\frac{s(s+1)}{2}}\delta_2\\
\no &\leq \delta_2 4^{\frac{s(s+1)}{2}}{n^{-1-\frac{\alpha s(s+1)}{8}}}.\ \
\end{align}
Similarly, for any $\pi_s \in \Pi_s$,
\begin{align}
\no \mathbb{P}\left(\big{|}\textbf{P}^{(1)} - \textbf{P}_{\pi_s}^{(l)}\big{|} \leq 2 \Delta_n\right) \leq \delta_2 4^{\frac{s(s+1)}{2}}{n^{-1-\frac{\alpha s(s+1)}{8}}}.\ \
\end{align}
Thus, as $n \to \infty$, the union bound yields:
\begin{align}
\no &\mathbb{P}\left(\bigcup\limits_{l=2}^{\frac{n}{s}} \left\{ \bigcup\limits_{{\pi_s} \in \Pi_s} \left\{ \big{|}\textbf{P}^{(1)} - \textbf{P}^{(l)}_{\pi_s} \big{|} \leq 2 \Delta_n\right\} \right\}\right) \leq \frac{n}{s}s!\delta_24^{\frac{s(s+1)}{2}}{n^{-1-\frac{\alpha s(s+1)}{8}}}\\
\no &\hspace{1.8 in} =(s-1)!4^{\frac{s(s+1)}{2}}{n^{-\frac{\alpha s(s+1)}{8}}}\to 0.\ \
\end{align}

Thus, with high probability, the distance between each of the $\textbf{P}^{(l)}$'s and $\textbf{P}^{(1)}$ is larger than $2\Delta_n$. Next, we show that, with high probability, each of the $\widetilde{\textbf{P}^{(l)}}$'s is significantly close to $\textbf{P}^{(l)}$. By using the union bound with (\ref{eq10}) and (\ref{eq11}), for $\zeta > \lceil{\frac{4}{\alpha}}\rceil$,
\begin{align}
\no \no &\mathbb{P}\left(\bigcup\limits_{l=2}^{\frac{n}{s}} \left\{ \mathcal{D}\left(\textbf{P}^{(l)},\widetilde{\textbf{P}^{(l)}}\right)\geq \Delta_n \right\}\right) \leq \sum\limits_{l=2}^{\frac{n}{s}} \mathbb{P}\left(\mathcal{D}\left(\textbf{P}^{(l)},\widetilde{\textbf{P}^{(l)}}\right)\geq \Delta_n\right)\\
\no &\leq ne^{-\frac{1}{2\sigma^2}n^{\frac{\alpha}{2}}} +\frac{s-1}{2}\tau \mathbb{E}\left[\left(X_{u}(k)X_{u'}(k)-\mu_{uu'}\right)^{\zeta}\right]{n^{1-\frac{\alpha}{4}\zeta}}\to 0,\ \
\end{align}
as $n \to \infty.$

Consequently, for all $l=2, 3, \cdots, \frac{n}{s}$, $\widetilde{\textbf{P}^{(l)}}$'s are close to $\textbf{P}^{(l)}$'s; thus, for large enough $n$,
\begin{align}
\no \mathbb{P}\left(\bigcup\limits_{l=2}^{\frac{n}{s}} \left\{ \mathcal{D}\left(\textbf{P}^{(1)},\widetilde{\textbf{P}^{(l)}}\right)\leq \Delta_n \right\}\right)\to 0. \ \
\end{align}
Hence, the adversary can successfully identify group $1$ among all of the groups with arbitrarily small error probability. 
\end{proof}

Finally, we show that the adversary can identify all of the members of group $1$ with arbitrarily small error probability.

\begin{lem}
If for any $\alpha>0$, the adversary obtains $m=n^{\frac{4}{s(s+1)}+\alpha}$ anonymized observations, and group $1$ is identified among all the groups, the adversary can identify user $1$ with arbitrarily small error probability.
\label{lemma3}
\end{lem}

\begin{proof}
Define sets $\mathcal{B}^{(n)}$ and $\mathcal{C}^{(n)}$ around $\mu_1$:
\begin{align}
\nonumber \mathcal{B}^{(n)}&= \left\{x \in (0,1); |x-\mu_1| \leq \Delta_n\right\},\\
\nonumber \mathcal{C}^{(n)}&=\left\{x\in (0,1); |x-\mu_1| \leq 2\Delta_n\right\},
\end{align}
where $\Delta_n = {n^{-\frac{2}{s(s+1)}-\frac{\alpha}{4}}}.$

Next, we show that when $m =cn^{\frac{4}{s(s+1)} + \alpha}$ and $n \to \infty$,
\begin{itemize}
	\item $\mathbb{P}\left( \big{|}\frac{S_1}{m}-\mu_{1}\big{|}\leq \Delta_n \right) \to 1.$
	\item $\mathbb{P}\left( \bigcup\limits_{u=2}^{s} \left\{\big{|}\frac{S_u}{m}-\mu_{1}\big{|}\leq \Delta_n\right\}\right) \to 0.$
\end{itemize}
In other words, the adversary examines $\widetilde{\mu_u}$'s which are defined according to (\ref{p-tilde2}) and chooses the only one that belongs to $\mathcal{B}^{(n)}$.

Substituting $\theta= \Delta_n$ in (\ref{A_u}) yields:
\begin{align}
\no \mathbb{P}\left( \big{|}\frac{S_1}{m}-\mu_{1}\big{|}\leq \Delta_n \right) &\geq 1- e^{-\frac{m\Delta_n^{2}}{2\sigma^2}}\\
\no &= 1-e^{-\frac{1}{2\sigma^2}n^{\frac{\alpha}{2}}} \to 1.
\end{align}
Thus, for large enough $n$,
\[\mathbb{P}\left({\widetilde{\mu_{\Pi(1) }}} \in \mathcal{B}^{(n)}\right) \to 1.\]

Next, we show that when $n \to \infty$,
$$\mathbb{P}\left( \bigcup\limits_{u=2}^{s} \left\{\big{|}\frac{S_u}{m}-\mu_{1}\big{|}\leq \Delta_n\right\}\right) \to 0.$$
By (\ref{case1}),
\[ \mathbb{P}\left( \mu_u\in \mathcal{C}^{(n)}\right) < 4 \Delta_n \delta_2.\]
Therefore, the union bound yields:
\begin{align}
\no \mathbb{P}\left( \bigcup\limits_{u=2}^s \left\{\mu_u \in \mathcal{C}^{(n)}\right\} \right) &\leq \sum\limits_{u=2}^s \mathbb{P}\left( \mu_u\in \mathcal{C}^{(n)}\right)\\
\nonumber &\leq 4s \Delta_n \delta_2\\
\no &\leq 4s n^{-\frac{2}{s(s+1)}-{\frac{\alpha}{4}}} \delta_2 \to 0.
\end{align}
Consequently, all $\mu_u$'s are outside of $\mathcal{C}^{(n)}$ with high probability. Next, we prove $\mathbb{P}\left({\widetilde{\mu_{\Pi (u)}}} \in \mathcal{B}^{(n)}\right)$ is small. Observe:
\begin{align}
\no \mathbb{P}\left( \big{|}\frac{S_u}{m}-\mu_{1}\big{|}\leq \Delta_n\right) &= \mathbb{P}\left( \big{|}\frac{S_u}{m}-\mu_{u}\big{|}\geq \Delta_n\right)\\
\no &\leq e^{-\frac{m\Delta_n^{2}}{2\sigma^2}}= e^{-\frac{1}{2\sigma^2}n^{\frac{\alpha}{2}}};\ \
\end{align}
hence, by the union bound, when $n$ is large enough,
\begin{align}
\no \mathbb{P}\left( \bigcup\limits_{u=2}^{s} \left\{\big{|}\frac{S_u}{m}-\mu_{1}\big{|}\leq \Delta_n\right\}\right) \leq se^{-\frac{1}{2\sigma^2}n^{\frac{\alpha}{2}}} \to 0.
\end{align}
Thus, if $m=n^{\frac{4}{s(s+1)}+\alpha}$, there exists an algorithm for the adversary to successfully identify user $1$ among all the users. 
\end{proof}

Next, we present Theorem~\ref{two_state_thm} which follows from Lemmas~\ref{lemma1},~\ref{lemma2}, and~\ref{lemma3}. In this theorem, we determine the required number of observations per user ($m$) for the adversary to break the privacy of each user, in terms of the number of users $(n)$ and size of group to which the user of interest belongs ($s$).

\begin{thm}\label{two_state_thm}
If the adversary knows both the structure of the association graph and the correlation coefficient between users, and $m=cn^{\frac{4}{s(s+1)}+\alpha}$, for any $\alpha > 0$; then, user $1$ has no privacy at time $k$. 
\end{thm}

Lastly, in Theorem~\ref{association_graph}, we consider the case where the adversary knows only the association graph, but not necessarily the correlation coefficients between the users. Similar to the arguments leading to Theorem~\ref{two_state_thm} and~\cite[Theorem 1]{ISIT18-longversion} we show that if $m$ is significantly larger than $n^{\frac{2}{s}}$, then the adversary can successfully break the privacy of the user of interest, i.e., the adversary can find an algorithm to estimate the actual data points of the user with vanishing small error probability. 

\begin{thm}\label{association_graph}
	If the adversary knows the structure of the association graph, and $m=cn^{\frac{2}{s}+\alpha'}$, for any $\alpha' > 0$;
	then, user $1$ has no privacy at time $k$. 
\end{thm}

\section{Discussion}
\label{discussion}
Here, we compare our results with previous work. When the users are independent, the adversary can break the privacy of each user if the number of the adversary's observations per user is $m=n^{2}$ \cite{KeConferance2017} (Case 1 in Figure~\ref{fig:converse2}). However, when the users are dependent, and the adversary knows their association graph (and not the correlation coefficients), each user will have no privacy if $m=n^{\frac{2}{s}}$ (Theorem \ref{association_graph}: Case 2 in Figure~\ref{fig:converse2}). The required number of per-user observations for the adversary to break the privacy of each user reduces further ($m=n^{\frac{4}{s(s+1)}}$) when the adversary has more information: the correlation coefficients between users (Theorem \ref{two_state_thm}: Case 3 in Figure~\ref{fig:converse2}). 
In other words, \emph{the more the adversary knows, the smaller $m$ must be}. Note that smaller $m$ means rapid changes in pseudonyms, which reduces the utility.
We have characterized the significance in the loss of privacy of various degrees of knowledge of user dependencies in this paper.
\begin{figure}
	\centering
	\includegraphics[width=0.9\linewidth]{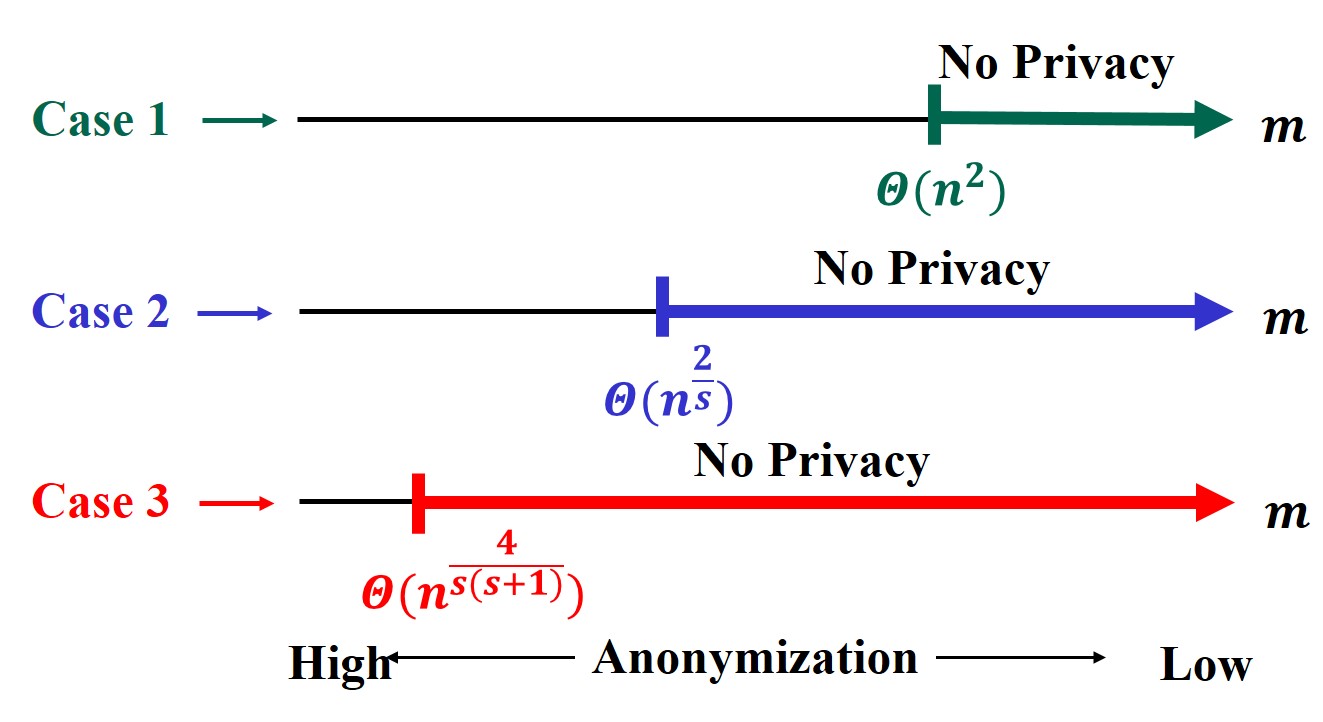}
	\centering
	\caption{Comparing the required number of observations per user for the adversary to break the privacy of each user for three cases: 1) independent users; 2) dependent users, adversary knows only the association graph; 3) dependent users, the adversary knows both the association graph and the correlation coefficient between users.}
	\label{fig:converse2}
\end{figure}

%% file: conclusion.tex
\section{Conclusion}
\label{conclusion}
Many popular applications use traces of user data, e.g., users' location information or medical records, to offer various services to the users. However, revealing user information to such applications put users' privacy at stake, as adversaries can infer sensitive private information about the users such as their behaviors, interests, and locations. In this paper, anonymization is employed to protect users' privacy when data traces of each user observed by the adversary are governed by i.i.d.\ Gaussian sequences, and data traces of different users are dependent. An association graph is employed to show the dependency between users, and both the structure of this association graph and the nature of the dependency between users are known to the adversary. We show that dependency is a significant detriment to the privacy of users. In comparison to the case in which data traces of different users are independent, here we must use a stronger anonymization technique by drastically increasing the rate at which user pseudonyms are changed, which degrades system utility.